\documentclass[conference]{IEEEtran}

\usepackage{amsmath,paralist,amsthm,comment,amssymb}
\usepackage{color,cite}

\usepackage{algorithm} 
\usepackage{algorithmic}
\usepackage{graphicx}
\usepackage{subfigure}

\newcommand{\mc}[1]{\mathcal{#1}}

\newcommand{\nix}[1]{}

\newtheorem{theorem}{Theorem}

\newtheorem{lemma}[theorem]{Lemma}

\ifCLASSINFOpdf
\else
\fi
\begin{document}
%
\title{
Equivalence of 2D color codes (without translational symmetry) 
to surface codes
}

\author{\IEEEauthorblockN{Arjun Bhagoji}
\IEEEauthorblockA{Department  of Electrical  Engineering\\
Indian Institute of Technology Madras\\
Chennai  600 036 India \\
Email: arjunbhagoji@gmail.com}
\and
\IEEEauthorblockN{Pradeep  Sarvepalli}
\IEEEauthorblockA{Department  of Electrical  Engineering\\
Indian Institute of Technology Madras\\
Chennai  600 036 India \\
Email: pradeep@ee.iitm.ac.in}
}


%


\maketitle

\begin{abstract}
In a recent work, Bombin, Duclos-Cianci, and Poulin showed that every local translationally invariant 2D  topological stabilizer code is locally equivalent to a finite number of copies of Kitaev's toric code. For 2D color codes, Delfosse relaxed the constraint on translation invariance and mapped a 2D color code onto three surface codes. In this paper, we propose an alternate map based on linear algebra.  We show that any 2D color code can be mapped onto exactly two copies of a related surface code. The surface code in our map is induced by the color code and easily derived from the color code. Furthermore, our map does not require any ancilla qubits for the surface codes. 

\end{abstract}



%
\IEEEpeerreviewmaketitle

\section{Introduction}

Toric codes \cite{kitaev03} proposed by Kitaev are one of the most studied classes of 
topological quantum codes  and of fundamental importance in fault tolerant quantum 
computing. Although toric codes and their generalization--surface codes, have many 
attractive features (such as local stabilizer generators, low complexity decoders, efficient fault tolerant protocols, a high circuit threshold  \cite{dennis2002,poulin10,raussen07}), they have a limited set of transversal gates. 
On other hand, a class of topological color codes can implement the entire Clifford 
group transversally \cite{bombin06}. This  might suggest that color codes are inequivalent  
to toric codes. However, in a very surprising development, Bombin et al. \cite{bombin2012} showed that
translationally invariant 2D color codes can be mapped to a finite number of copies of Kitaev's toric code. 

However, the result in \cite{bombin2012}, as well the subsequent papers \cite{yoshida11,bombin14} which explore the equivalence between stabilizer codes and 
toric codes in great detail, have one important qualifier, namely translational invariance.
For 2D color codes, Delfosse \cite{delfosse2014} relaxed the constraint on translation invariance and mapped a 2D color code to three surface codes. 
In this paper, we propose an alternate map based on linear algebra. 
We 
map arbitrary color codes, including those that are not translationally invariant, onto two copies of a surface code.

The main differences between the proposed map (and its variations) and that of \cite{bombin2012} are: first, the map therein requires local translation symmetry. Second, in the map in \cite{bombin2012}, the number of toric codes onto which the color code is mapped need not always be two. On the other hand, our map always gives exactly two surface codes. These surface codes need not be copies of the toric code on the square lattice. Third, the proposed map does not require any ancilla qubits, as may be the case for some codes under the map in \cite{bombin2012}. Even for arbitrary color codes, our map is efficiently computable locally and we compute the images for all the single qubit errors on the color code in closed form. 
On the other hand, 
the results in \cite{bombin2012} go beyond color codes and include all local translationally invariant 2D stabilizer codes and certain subsystem codes. 

Our work differs from that of \cite{delfosse2014} in the following aspects. 
The map in \cite{delfosse2014} projects  onto three copies of surface codes. 
Therefore, our map leads to  a lower decoding complexity compared to \cite{delfosse2014}.  Furthermore, unlike our map which is a bijective map onto the surface codes, the map in \cite{delfosse2014} is not bijective, although it is injective.  This has a bearing in the context of decoding.  In some cases, a decoder  using the map in \cite{delfosse2014} may not be able to lift an error (estimate) from surface codes to (the parent) color code. This will not occur with a decoder using our map.

Following the submission of our paper, we became aware of the work by Kubica, Yoshida, and Pastawski \cite{kubica15} who showed equivalence between color codes and toric codes for all dimensions $D\geq 2$. 
In 2D, for color codes without boundaries, their result is similar to ours but there are substantial differences.
First, they map the color code onto  two {\em different}  surface codes, we map 
onto two copies of the {\em same} surface code. Second, we use linear algebra to study these equivalences, which is simpler than the approach taken in \cite{kubica15} (or \cite{bombin2012,delfosse2014}).

One immediate application of our results, as in \cite{bombin2012,delfosse2014, kubica15}, is an alternate decoding scheme for color codes via surface codes. 

\section{Preliminaries}

We assume that the reader is familiar with stabilizer codes \cite{calderbank98,gottesman97} and topological codes \cite{kitaev03}. 
The Pauli group on $n$ qubits is denoted $\mc{P}_n$. 
We denote the vertices of a graph $\Gamma$ by $\mathsf{V}(\Gamma)$, and the edges by $\mathsf{E}(\Gamma)$. The set of edges incident on a vertex $v$
is denoted as $\delta(v)$ and the edges in the boundary of a face by $\partial(f)$. Assuming that $\Gamma$ is embedded on a suitable surface we use $\mathsf{F}(\Gamma)$ to denote the faces of the embedding and do not always make explicit reference to the surface.
A surface code on a graph $\Gamma$ is a stabilizer code where the qubits are placed on the edges of $\Gamma$ and whose stabilizer $S$ is given by 
\begin{align}
S &= \langle A_v , B_f \mid v\in \mathsf{V}(\Gamma), f\in \mathsf{F}(\Gamma) \rangle, 
\label{eq:stab-sc}
\end{align}
where $A_v = \prod_{e\in \delta(v) }X_e \mbox{ and  } B_f =\prod_{e\in \partial(f)}Z_e$.
The Pauli group on the qubits of a surface code is denoted as $\mc{P}_{\mathsf{E}(\Gamma)}$. 
A 2-colex is  a trivalent, 3-face-colorable complex. 
A stabilizer code is defined on a 2-colex by attaching qubits to every vertex  
and defining the stabilizer $S$ as 
\begin{eqnarray}
S = \langle B_f^X, B_f^Z \mid v\in \mathsf{F}(\Gamma) \rangle \mbox{ where } B_f^\sigma = \prod_{v\in f }\sigma_v.  \label{eq:stab-tcc}
\end{eqnarray}
We denote the Pauli group on these qubits as $\mc{P}_{\mathsf{V}(\Gamma)}$; the $c$-colored faces of  $\Gamma$ by $\mathsf{F}_c(\Gamma)$ and the $c$-colored edges of $\Gamma$ by $\mathsf{E}_c(\Gamma)$.
We restrict our attention to 2-colexes which do not have boundaries or multiple edges (the surface codes could contain multiple edges though). 
This is not a severe restriction because a 2-colex with multiple edges can be modified to another 2-colex without such edges but encoding the same number of qubits and possessing the same error correcting capabilities (in terms of distance). 
All embeddings are assumed to be 2-cell embeddings 
i.e. faces are homeomorphic to unit discs. 

There are four types of topological charges on a surface code: i) electric charge (denoted $\epsilon$) localized on the vertices, ii) magnetic charge (denoted $\mu$) living on the plaquettes,  iii) the composite electric and magnetic charge denoted $\epsilon \mu$ which resides  on both the plaquettes and vertices, and iv) the vacuum denoted $\iota$. Of these, only two charges are independent. We shall take this pair to be the electric and magnetic charges. 
A charge composed with another charge of the same type gives the vacuum i.e. $c\times c =\iota$. The electric charges are created by $Z$-type errors and magnetic charges by $X$-type errors on the surface code.  

On a color code, the topological charges live on the  faces. In addition to being  electric and/or magnetic, they also carry a color depending on which face they are present. Let us denote the electric charge on a $c$-colored face  as $\epsilon_c$, the magnetic charge as $\mu_c$ and the composite charge as $\epsilon_c\mu_{c}$. The electric charges are not all independent \cite{bombin06}. Any pair (two out of three colors) of them can be taken as the independent set of electric charges. Similarly, only two magnetic charges are independent. As for surface codes, electric (magnetic) charges are created by $Z$ ($X$) errors on the color code.

A hopping operator is any element of the Pauli group that moves the charges. On a surface code, we can move the electric charges from one vertex to another by means of a $Z$-type Pauli operator. We denote by $H_{u\leftrightarrow v}^{\epsilon}$ the operator that moves $\epsilon$ from vertex $u$ to $v$ and vice versa.
If we consider the magnetic charges, then the movement can be accomplished by means of an $X$-type Pauli operator. The operator that moves a magnetic charge from face $f$ to $f'$ (or vice versa) is denoted by $H_{f\leftrightarrow f'}^{\mu}$.
Elementary hopping operators are those which move charges from one vertex to an adjacent vertex or from one plaquette to an adjacent plaquette.  Let $e=(u,v)$ be the edge incident on the vertices $u$, $v$. We denote the elementary hopping operator along $e$ as $H_e^{\epsilon}$, where $ H_e^{\epsilon}  = Z_{e}$. It is a specific realization of $H_{u\leftrightarrow v}^{\epsilon}$. Similarly, the elementary operator that moves $\mu $ across $e$ is denoted as $H_e^\mu$. Let $e$ be the edge shared by the faces  $f$ and $f'$, then $H_{f\leftrightarrow f'}^{\mu}$ can be realized by $H_e^{\mu }$ where $H_e^{\mu }=X_{e}$. Observe that $H_{u\leftrightarrow v}^{\epsilon}$ and $H_{f\leftrightarrow f'}^{\mu}$ anti-commute when they act along the same edge, while operators for the same type of charges commute. In general, $H_{u\leftrightarrow v}^{\epsilon}$ and $H_{f\leftrightarrow f'}^{\mu}$ commute if and only if they cross an even  number of times. 

Similarly, we can define hopping operators for  color codes. Let $f, f' \in \mathsf{F}_c(\Gamma)$ be two  plaquettes connected by an edge $(u,v)$ where $u\in f$
and $v\in f'$. Then $H_{f\leftrightarrow f'}^{\epsilon_c}$ and $H_{f\leftrightarrow f'}^{\mu_c}$ are the operators that move $\epsilon_c$ and $\mu_c$ from $f$ to $f'$. A realization of these operators  along $(u,v)$ is $H_{u,v}^{\epsilon_c} =  Z_{u}Z_{v} $ and  $H_{u,v}^{\mu_c}=X_{u}X_{v}$.
An element of the stabilizer can be viewed as a combination of hopping operators which move a charge around and bring it back to the original location. Since this movement cannot be detected, we can always adjoin an element of the stabilizer to the hopping operators.

\section{Mapping a color code to  two copies of a surface code}

\subsection{Color codes to surface codes---Constraints}\label{ssec:constraints}
Our goal is to find a map  between a color code and some related surface codes. We shall denote this map by $\pi$ for the rest of the paper. 
We shall first describe the construction of $\pi$ in an informal fashion, emphasizing the principles underlying the map, and then rigorously justify all the steps. 
The key observation, due to \cite{bombin2012}, is that there are four types of charges on a surface code and sixteen types of charges on a color code. This is the starting point for relating the color code to surface codes. The two pairs of independent charges on the color code i.e. $\{\epsilon_c, \mu_{c'} \}$ and $\{\epsilon_{c'}, \mu_{c} \}$ suggest that we can decompose the color code into a pair of toric codes by mapping  $\{\epsilon_c, \mu_{c'} \}$ charges onto one toric code and $\{\epsilon_{c'}, \mu_{c} \}$ onto another. 
 However, charge ``conservation'' is not the only constraint. 
We would like a map that preserves in some sense the structure of the color code and allows us to go back and forth between the color code and the surface codes. We shall impose some conditions on this map keeping in mind that we would like to use it in the context of decoding  color codes. 

First, observe that the electric charges on the surface codes live on the vertices while the magnetic charges live on the plaquettes. But, if we consider the pair of charges 
$\{\epsilon_c, \mu_{c'} \}$, they both live on plaquettes---one on the $c$-colored plaquettes and another on $c'$-colored 
plaquettes. A natural way to make the association to a  surface code is to contract all the $c$-colored plaquettes in the embedding of $\Gamma$. This will give rise to a new graph $\tau_c(\Gamma)$. We can now place the charges $\epsilon_c$ and $\mu_{c'}$ on the vertices and plaquettes of $\tau_c(\Gamma)$ respectively. Similarly, the charges 
$\{\mu_{c}, \epsilon_{c'} \}$ can live on the vertices and plaquettes of {\em another} instance of $\tau_c(\Gamma)$.
We impose the following (desirable) constraints on the map $\pi$. It must be 
{(i) linear,
 (ii) invertible,
 (iii) local,
 (iv) efficiently computable,
 (v) preserve the commutation relations between the (Pauli) error operators on $\mathsf{V}(\Gamma)$ i.e. $\mc{P}_{\mathsf{V}(\Gamma)}$,
 and (vi) consistent in the description of the movement of charges on the color code and surface codes.
}These constraints are not necessarily independent and in no particular order. It is possible to relax some of the constraints above. 

\subsection{Deducing the map---A linear algebraic approach }

The maps proposed in \cite{bombin2012} are based on the following ideas: i) conservation of
topological charges ii) identification of the hopping operators and iii) preserving the commutation relations between the hopping operators.
These ideas are central to our work as well. However, we take a simpler linear algebraic approach 
to find the map.

Suppose we have  a 2-colex $\Gamma$. Then, upon contracting all the $c$-colored faces including their boundary edges, we obtain another complex. We denote this operation as $\tau_c$ and the resulting complex as $\tau_c (\Gamma)$ (see Fig.~\ref{fig:tau}). We suppress the subscript if the context makes it clear and just write $\tau$. 
There is a one-to-one correspondence between the $c$-colored faces of $\Gamma$ and the vertices of $\tau(\Gamma)$, so we can label the vertices of $\tau(\Gamma)$ by $f\in \mathsf{F}_c(\Gamma)$. We also label them by $\tau(f)$ to indicate that the vertex was obtained by contracting $f$. Similarly, the edges of $\tau(\Gamma)$ are in one-to-one correspondence with the $c$-colored edges of $\Gamma$, so an edge $\tau(\Gamma)$ is labeled the same as the parent edge $e=(u,v)$ in $\Gamma$.  
The faces which are not in $\mathsf{F}_c(\Gamma)$ are mapped to faces of $\tau(\Gamma)$. Therefore, we label the faces as 
$f$ or more explicitly as $\tau(f)$, where $f\not\in  \mathsf{F}_c(\Gamma) $. Thus, the complex $\tau(\Gamma)$ has the vertex set $\mathsf{F}_c(\Gamma)$,
edge set $\mathsf{E}_c(\Gamma)$ and faces $\mathsf{F}_{c'}(\Gamma)\cup \mathsf{F}_{c''}(\Gamma)$. Since every vertex $v$ in 
$\Gamma$ has a unique $c$-colored edge incident on it, we can associate to it an edge in $\tau(\Gamma)$ as $\tau(v)$. 

\begin{figure}[htb]
\centering
\subfigure{
\centering

\includegraphics[scale=0.5]{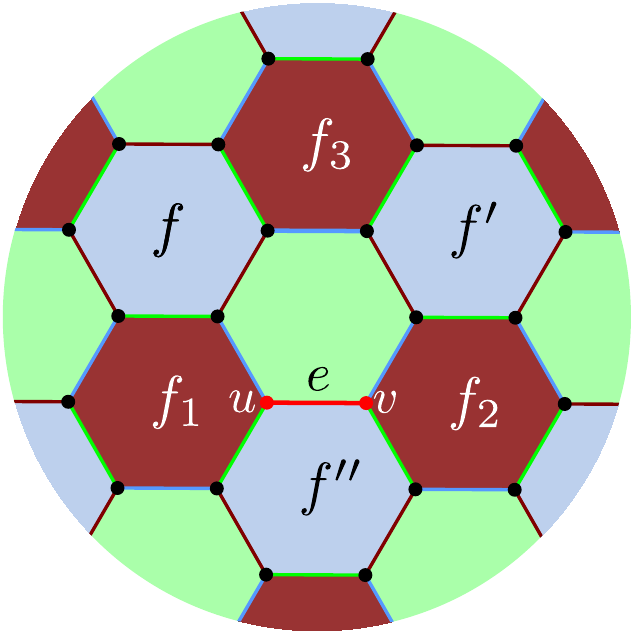}\label{fig:tcc}

}
\subfigure{
\centering
\includegraphics[scale=0.5]{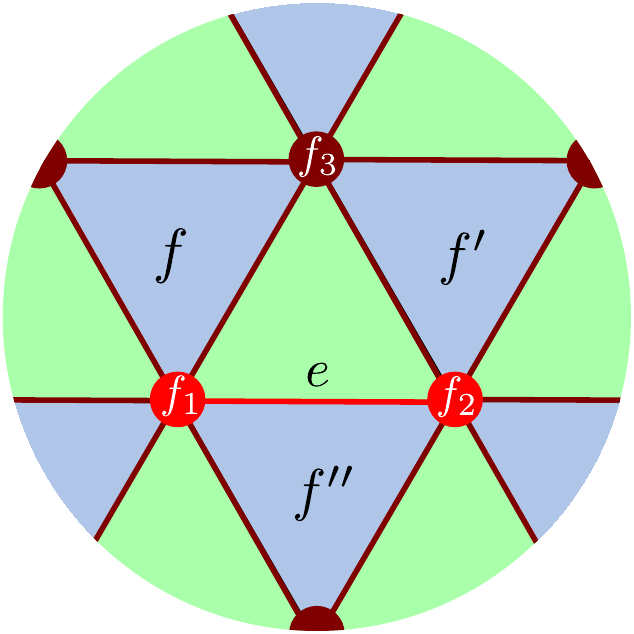}\label{fig:sc}
}
\caption{Illustrating the contraction of a color code via $\tau_c$ and the resultant surface code. Only portions of the  codes are shown. The $c$-colored faces are vertices in $\tau_c(\Gamma)$. The faces $f \not\in \mathsf{F}_c(\Gamma)$ remain faces in $\tau_c(\Gamma)$ and are also labeled $f$ in $\tau_c(\Gamma)$, while the $c$-colored  edge $e=(u,v)$ in $\Gamma$ is mapped to an edge in $\tau_c(\Gamma)$, so we retain the label $e$. Every vertex in $\Gamma$ is incident on a unique $c$-colored edge, so we can also extend $\tau_c$ to vertices $u$, $v$ and edges unambiguously by defining $\tau_c(u)=\tau_c(v)=\tau_c(u,v)=e$.}\label{fig:tau}
\end{figure}

Now, each $c$-colored face in $\Gamma$ can host $\epsilon_c$ and $\mu_c$. With respect to $\tau_c(\Gamma)$, they both reside on the vertices of $\Gamma_c$. So we shall place them on two different copies of $\tau_c(\Gamma$) denoted $\Gamma_1$ and $\Gamma_2$. 
Then, the charges $\epsilon_c$ and $\mu_c$ will play the role of an electric charge on $\Gamma_1$ and $\Gamma_2$, respectively. So, we shall make the identification $\epsilon_c \equiv \epsilon_1$ and $\mu_c\equiv \epsilon_2$. The associated magnetic charges on $\Gamma_i$ will have to reside on $\mathsf{F}(\Gamma_i)$. Possible candidates for these charges must come from $\epsilon_{c'}$, $\epsilon_{c''}$ and  $\mu_{c'}$, $\mu_{c''}$. The following lemma addresses these choices. 
 
\begin{lemma}[Charge mapping] \label{lm:charge-pairing}
Let $c,c',c''$ be three distinct colors.  Then, $\{\epsilon_c,\mu_{c'}\}$ and $\{ \epsilon_{c'},\mu_c \}$ are permissible pairings of the charges so that the color code on $\Gamma$ can be mapped to a pair of surface codes on $\Gamma_i=\tau_c(\Gamma)$. In other words, $\epsilon_1\equiv \epsilon_c$, $\mu_1\equiv \mu_{c'}$, $\epsilon_2\equiv \mu_c$ and $\mu_1\equiv \epsilon_{c'}$, where $\epsilon_i$ and $\mu_i$ are the electric and magnetic charges of the surface code on $\Gamma_i$.
 
\end{lemma}
\begin{proof} First, observe that operators that move the electric charges $\epsilon_c$ and $\epsilon_{c'}$ are both $Z$-type,
therefore they will always commute. This means that if $\epsilon_c$ is identified with the electric charge on a surface code,
$\epsilon_{c'}$ cannot be the associated magnetic charge. That leaves either $\mu_{c}$ and $\mu_{c'}$. Of these, observe 
that any operator that moves $\mu_{c}$ will always overlap with any operator that moves $\epsilon_c$ an even number of times. 
Therefore, this leaves only $\mu_{c'}$. The operators that move $\epsilon_c$ and $\mu_{c'}$ commute/anti-commute when they overlap an even/odd number of times just as the electric and magnetic charges of a surface code justifying the association $\epsilon_1\equiv \epsilon_c$ and $\mu_1\equiv \mu_{c'}$. A similar argument shows the validity of the equivalence $\epsilon_2\equiv \mu_c$ and $\mu_2\equiv \epsilon_{c'}$. 
\end{proof}

Let $\Gamma$ have $n$ vertices and $F_c$ vertices of color $c$. Then, $\Gamma_i$ has $F_c$ vertices, $n/2$ edges and $F_{c'}+F_{c''}$ faces. Together $\Gamma_1$ and $\Gamma_2$ have $n$ qubits. We desire that $\pi$ accurately reflect the movement of the independent charges on the color code and the surface codes. So, $\pi$ must map the hopping operators of the charges of the color code on  $\Gamma$ to the hopping operators of the surface code on $\Gamma_i$. As mentioned earlier,
$H_{u,v}^{\epsilon_c}$ moves electric charges on $c$-colored plaquettes and 
$H_{u,v}^{\epsilon_{c'}}$ electric charges on $c'$-colored plaquettes. But, although these charges may appear to be independent, due to the structure of the color code they are not.
A $c''$-colored plaquette on the color code is bounded by edges whose color alternates between $c$ and $c'$. 
The $Z$-type stabilizer associated to this plaquette, i.e. $B_f^Z$, can be viewed as being composed of $H_{u,v}^{\epsilon_c}$ hopping operators that move
$\epsilon_c$, in which case we would expect to map $B_f^Z$ onto $\Gamma_1$. But, $B_f^Z$  can also be viewed as 
being composed of $H_{i,j}^{\epsilon_{c'}}$. Thus, we see that there are two possible combinations of hopping operators that give the same plaquette stabilizer; one composed entirely of hopping operators of $c$-colored charges and the other of hopping operators of $c'$-colored charges. This suggests that there are dependencies among the hopping operators and some of them, while ostensibly acting on only one kind of charge, could still be moving the other type of charges. However, the overall effect on the other charge must be trivial, i.e. it must move the charge back to where it started. A similar argument can be made for $B_f^X$ which moves the magnetic charges. The next lemma makes precise these dependencies. 

\begin{lemma}[Dependent hopping operators]\label{lm:dep-ops}
Let $f \in \mathsf{F}_{c''}(\Gamma)$ and ${1},\ldots, {2{\ell_f}}$ be the vertices in its boundary so that 
$({2i-1},{2i}) \in \mathsf{E}_c(\Gamma)$,  $({2i},{2i+1}) \in \mathsf{E}_{c'}(\Gamma)$ 
for $1\leq i\leq {\ell_f}$ and $2{\ell_f}+1 \equiv 1$. If $\pi$ is invertible, then
 $\pi(B_f^\sigma)\neq I$ and there are $4{\ell_f}-2$ independent 
elementary hopping operators along the edges of $f$.
\end{lemma}

\begin{proof}
The stabilizer generator $B_f^Z$ is given as 
\begin{align}
B_f^Z& = \prod_{i=1}^{2{\ell_f}}Z_i  = \prod_{i=1}^{\ell_f} Z_{2i-1} Z_{2i}=Z_1Z_{2{\ell_f}}\prod_{i=1}^{{\ell_f}-1} Z_{2i} Z_{2i+1}\\
&=\prod_{i=1}^{\ell_f} H_{2i-1, 2i}^{\epsilon_c} =H_{1, 2{\ell_f}}^{\epsilon_{c'}}\prod_{i=1}^{{\ell_f}-1} H_{2i, 2i+1}^{\epsilon_{c'}}.
\end{align}
We see that $B_f^Z$ can be expressed as the product of ${\ell_f}$ hopping operators of type $H_{u,v}^{\epsilon_c}$ or type $H_{u,v}^{\epsilon_{c'}}$. Further, we have
\begin{eqnarray}
\pi(B_f^Z)& =&\prod_{i=1}^{\ell_f} \pi(H_{2i-1, 2i}^{\epsilon_c}) =\pi( H_{1, 2{\ell_f}}^{\epsilon_{c'}})\prod_{i=1}^{{\ell_f}-1} \pi(H_{2i, 2i+1}^{\epsilon_{c'}})\nonumber
\end{eqnarray}

If $\pi(B_f^Z)=I$, then $\ker(\pi)\neq I$ which means that $\pi $ is not invertible and it would not be possible to preserve the information about the syndromes, as 
$\pi(B^Z_f)$ would commute with all the error operators. So, we require that $\pi(B_f^Z)\neq I$. 
This means that only one of these hopping operators is dependent and there are $2{\ell_f}-1$ independent hopping operators. The linear independence 
of the remaining $2{\ell_f}-1$ operators can be easily verified by considering their support.
Similarly, $B_f^X$ also implies that there are another $2{\ell_f}-1$ independent hopping operators, giving us $4{\ell_f}-2$ in total. 
\end{proof}

We are now ready to define the action of $\pi$ on  elementary hopping operators. Without loss of generality we can assume if $f\in \mathsf{F}_{c''}(\Gamma)$  has $2{\ell_f}$ edges, then the dependent hopping operators of $f$ are $H_{1,2{\ell_f} }^{\epsilon_{c'}}$  and 
$H_{2m,2m+1 }^{\mu_{c'}}$ i.e. $Z_1Z_{2{\ell_f}}$ and  $X_{2m}X_{2m+1}$, where $1\leq m\leq {\ell_f}$ and $2{\ell_f}+1\equiv 1$. 

\begin{lemma}[Elementary hopping operators]\label{lm:hopping}
Let $f, f' \in \mathsf{F}_c(\Gamma)$ where the edge $(u,v)$ is incident on $f$ and $f'$. 
Then, the following choices reflect the charge movement on $\Gamma$ onto the surface codes on $\Gamma_i$.
\begin{eqnarray}
 \pi(H_{u, v}^{\epsilon_c})  & = & \left[Z_{\tau(u)}\right]_1  = \left[Z_{\tau(v)}\right]_1 \label{eq:elec-hopper1}\\
 \pi(H_{u , v}^{\mu_c}) & = & \left[Z_{\tau(u)}  \right]_2= \left[Z_{\tau(v)}\right]_2,\label{eq:mag-hopper1}
 \end{eqnarray}
 where $[ T  ]_i$  indicates the instance of the surface code on which $T$ acts.
 Now if $f, f' \in \mathsf{F}_{c'}(\Gamma)$ and $(u,v)\in \mathsf{E}_{c'}(\Gamma)$ such that $u\in f$ and $v\in f'$ and
$H_{u,v}^{\epsilon_{c'}}$ and $H_{u,v}^{\mu_{c'}}$ are chosen to be independent hopping operators of $f$, then
\begin{eqnarray}
 \pi(H_{u, v}^{\epsilon_{c'}})  = \left[X_{\tau(u)} X_{\tau(v)}\right]_2  
 \mbox{; } \pi(H_{u , v}^{\mu_{c'}}) = \left[X_{\tau(u)} X_{\tau(v)} \right]_1. \label{eq:mag-hopper2}
\end{eqnarray}
 
\end{lemma}
\begin{proof}
We only prove for $H_{u, v}^{\epsilon_c}$ and $H_{u , v}^{\mu_{c'}}$. Similar reasoning can be employed for $H_{u,v}^{\mu_c}$
and  $H_{u, v}^{\epsilon_{c'}}$.
(i)  $H_{u, v}^{\epsilon_c}$: This operator moves $\epsilon_c$ from $f$ to $f'$ in $\Gamma$. These 
faces are mapped to adjacent vertices in $\tau(\Gamma)$. By Lemma~\ref{lm:charge-pairing}, $\epsilon_c$ is mapped to $\epsilon_1$, so $\pi(H_{u, v}^{\epsilon})$ should move $\epsilon_1$ from the vertex $\tau(f)$ to the vertex $\tau(f')$ on $\Gamma_{1}$. 
Many hopping operators can achieve this; choosing the elementary operator gives $\pi(Z_u Z_v)= [Z_{\tau(u,v)}]_1$. Since
$\tau(u,v)=\tau(u)=\tau(v)$, Eq.~\eqref{eq:elec-hopper1} follows.
(ii) $H_{u, v}^{\mu_{c'}}$: This operator moves $\mu_{c'}$ from $f$ to $f'$. Since $\mu_{c'}$ is mapped to $\mu_1$, $ \pi(H_{u, v}^{\mu_{c'}})$ should move $\mu_1$ from the plaquette $\tau(f)$ to $\tau(f')$ on $\Gamma_{1}$. The operator on the first surface code which achieves this is an $X$-type operator on qubits $\tau(u)$ and $\tau(v)$ in $\Gamma_1$, i.e. $[X_{\tau(u)}X_{\tau(v)}]_1$. In both cases we choose the hopping operators to be of minimum weight. 
\end{proof}

\begin{figure}
\centering
\subfigure{
\includegraphics[scale=0.25,angle=-90]{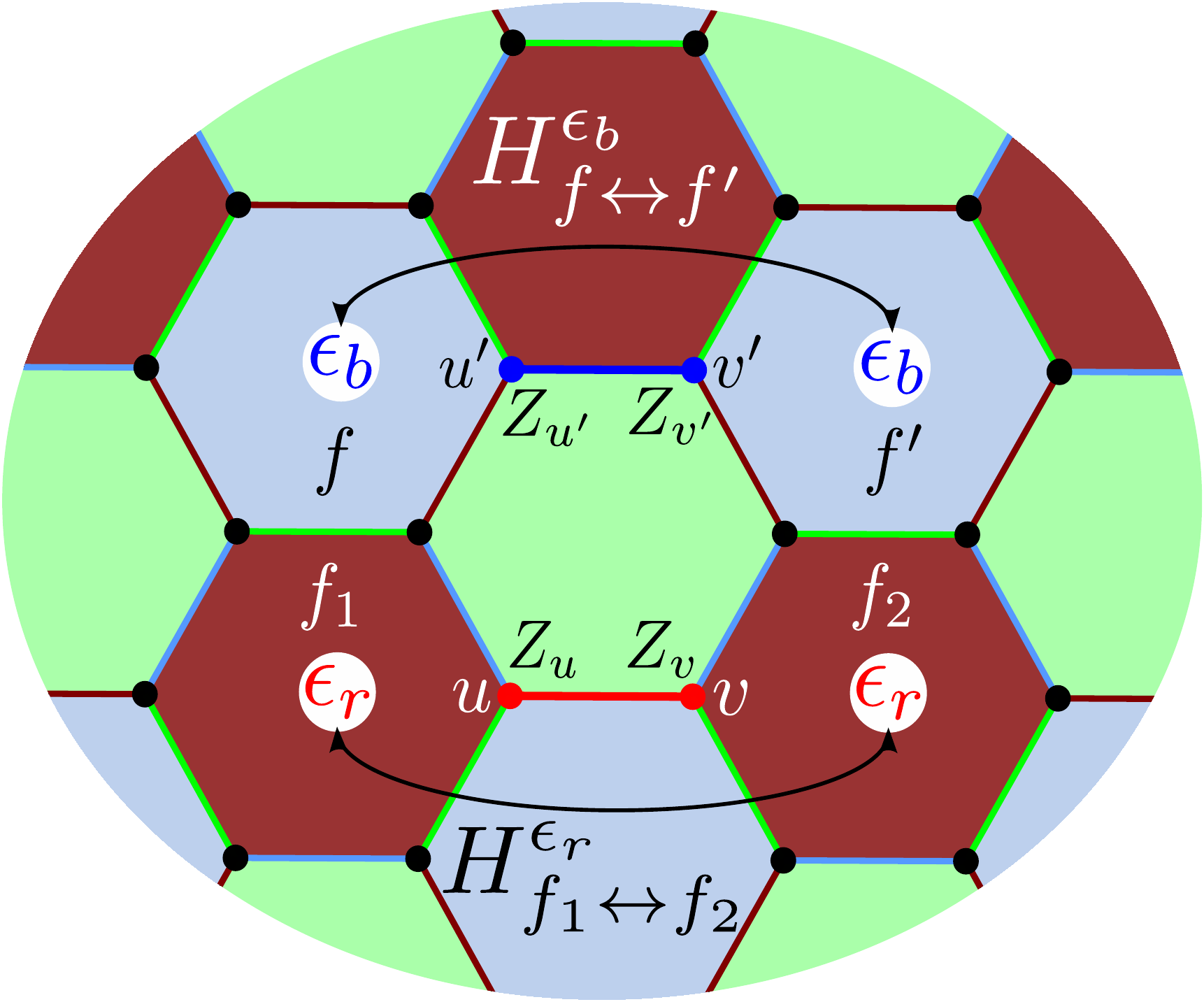}
}
\subfigure{
\includegraphics[scale=0.22
,angle=-90]{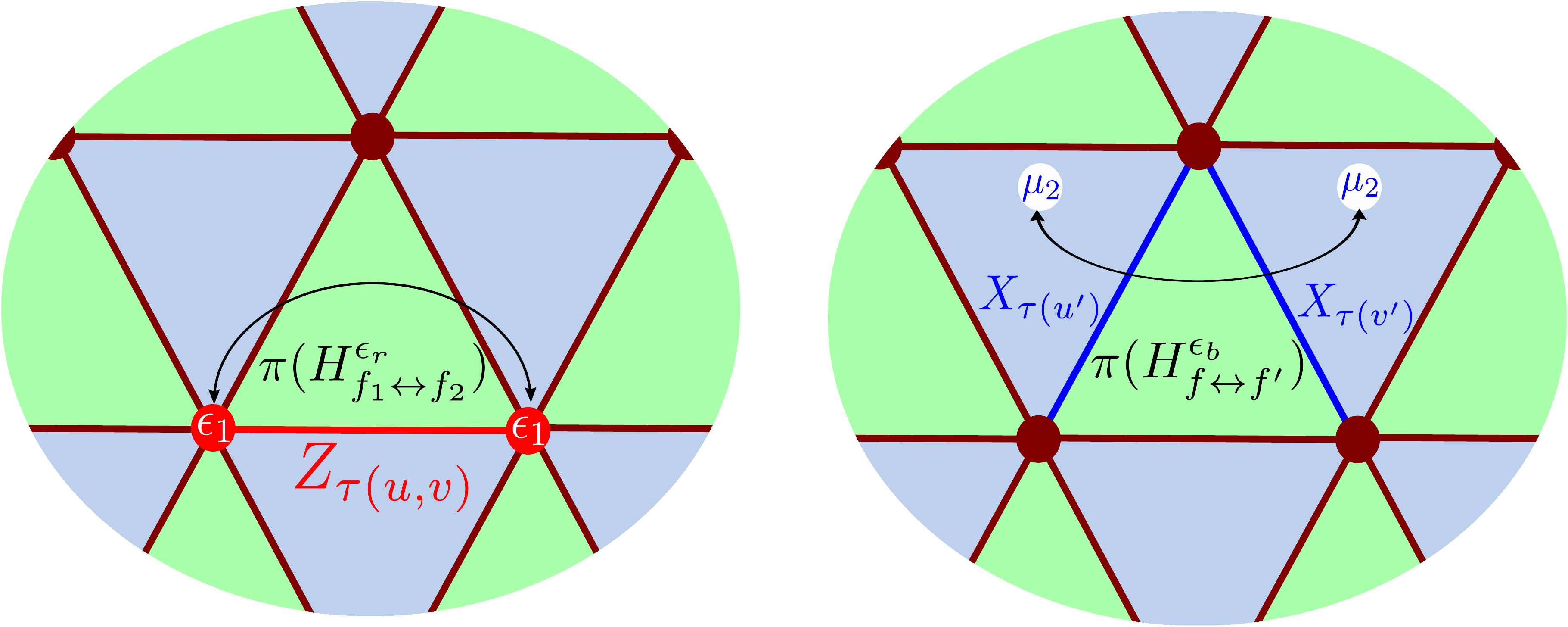}
}
\caption{Mapping the independent hopping operators $H_{u,v}^{\epsilon_r}=H_{f_1 \leftrightarrow f_2}^{\epsilon_r}=Z_u Z_v$ and 
$H_{u',v'}^{\epsilon_b}=H_{f \leftrightarrow f'}^{\epsilon_b}=Z_{u'}Z_{v'}$ on $\Gamma$ 
onto two copies of $\tau(\Gamma)$ i.e. $\Gamma_1$ and $\Gamma_2$; 
$\pi(H_{f_1 \leftrightarrow f_2}^{\epsilon_r}) = [Z_{\tau(u)}]_1 $ acts only on $\Gamma_1$ while 
$H_{f \leftrightarrow f'}^{\epsilon_b} = [X_{\tau(u')}X_{\tau(v')}]_2$ acts only on $\Gamma_2$.}
\end{figure}

Lemma~\ref{lm:hopping} does not specify the mapping for the dependent hopping operators but it can be  obtained as a linear combination of the independent ones. 
Alternative choices to those given in Lemma~\ref{lm:hopping} exist for $\pi$. These choices are essentially alternate hopping operators on the surface codes which accomplish the same charge movement. 
Such operators can be obtained by adding stabilizer elements to those given in 
Eqs.~\eqref{eq:elec-hopper1}--\eqref{eq:mag-hopper2}.

In this paper we explore the choice when the operators $H_{1,2{\ell_f} }^{\epsilon_{c'}}$  and $H_{2m,2m+1 }^{\mu_{c'}}$ are dependent. 
The $c''$-faces form a covering of all the vertices of $\Gamma$ and they are non-overlapping. 
The elementary hopping operators along the edges on such plaquette do not interact with the elementary hopping operators of other plaquettes in $\mathsf{F}_{c''}(\Gamma)$. So we can consider each $f\in \mathsf{F}_{c''}(\Gamma)$ independently. This also makes sense from our constraint to keep $\pi$ local. Based on Lemmas~\ref{lm:dep-ops}~and~\ref{lm:hopping}, we can map the independent elementary hopping operators of $f$ along $c$-colored edges. 
They map elementary hopping operators on $\Gamma$ to elementary  hopping operators on $\Gamma_i$.
\begin{align}
\pi (Z_{2i-1}Z_{2i} )  = [Z_{\tau(2i)}]_1 &\mbox{ and } 
\pi (X_{2i-1}X_{2i} ) = [Z_{\tau(2i)}]_2 \label{eq:cedge}
\end{align}

Next, we consider the hopping operators that involve the $c'$-colored edges. Without loss of generality we assume that the edge $Z_1 Z_{2{\ell_f}}$ is the one which carries the dependent hopping operator and $X_{2m}X_{2m+1}$ carries the other dependent hopping operator. Then letting $2{\ell_f}+1\equiv 1$ we have
\begin{align}
\pi(Z_{2i}Z_{2i+1}) &= [X_{\tau(2i)}X_{\tau(2i+1)}]_2 \mbox{ ; } 1\leq i < {\ell_f}  \label{eq:zz-diag}\\
\pi(X_{2i}X_{2i+1}) &= [X_{\tau(2i)}X_{\tau(2i+1)}]_1 \mbox{ ; } 1\leq i\neq m \leq {\ell_f}.
\label{eq:xx-diag}
\end{align}

All these operators and their images under $\pi$ are linearly independent as can be seen from  their supports. 
From Lemma~\ref{lm:hopping} we obtain the images for the dependent hopping operators:
\begin{align}
\pi(H_{1,2{\ell_f} }^{\epsilon_{c'}})&=  [X_{\tau(1)} X_{\tau(2{\ell_f})}]_2\prod_{i=1}^{\ell_f} [Z_{\tau(2i)}]_1 \\
\pi(H_{2m,2m+1 }^{\mu_{c'}}) &=  [X_{\tau(2m)}X_{\tau(2m+1)}]_1
\prod_{i=1}^{\ell_f} [Z_{\tau(2i)}]_2
\end{align}
To complete the map it remains to find the action of $\pi$ for two more independent errors on the color code. One choice is any pair of single qubit operators $X_i$ and $Z_j$, where $1\leq i,j\leq 2{\ell_f}$. Or we can consider the images under $\pi$. We can see
from Eqs.~\eqref{eq:cedge}--\eqref{eq:xx-diag} that the images are also linearly independent and only single qubit $X$-type of errors remain to be generated. One choice is any  $[X_{\tau(i)}]_1$ on $\Gamma_1$ and $[X_{\tau(j)}]_2$ on $\Gamma_2$, where $1\leq i,j\leq 2{\ell_f}$. That is, we need to  find $E,E'$ such that $\pi(E)=[X_{\tau(i)}]_1$ and $\pi(E')=[X_{\tau(j)}]_2$ respect the commutation relations. 
Lemma~\ref{lm:split} addresses  this choice.

\begin{lemma}[Splitting]\label{lm:split} 
The following choices lead to an invertible $\pi$ while respecting the commutation relations with hopping operators in Eqs.~\eqref{eq:cedge}--\eqref{eq:xx-diag}. 
\begin{align}
\pi(g X_1  ) &=[X_{\tau(1)}]_1 \mbox{ where } g\in \{I,B_f^X, B_f^Y, B_f^Z \}\label{eq:z-split}\\
\pi(g Z_{2m} ) &=[X_{\tau(2m)}]_2 \mbox{ where } g\in \{I,B_f^X \}\label{eq:x-split}
\end{align} 
\end{lemma}

\begin{proof}
Each face $f\in \mathsf{F}_{c''}(\Gamma)$ accounts for $2\ell_f$ qubits i.e. $4\ell_f$ independent operators.
Now $[X_{\tau(1)}]_1$ and $[X_{\tau(2m)}]_2$ form a linearly independent set of size $4{\ell_f}$ along with the images of the independent elementary hopping operators on $f$.
Thus, the elementary hopping operators and the preimages of $[X_{\tau(1)}]_1$ and $[X_{\tau(2m)}]_2$ 
account for all the $4\ell_f$ operators on qubits on $f$.
Considering all faces in $\mathsf{F}_{c''}(\Gamma)$,  we have $\sum_f 4\ell_f=2n$ operators
which generate $\mc{P}_{\mathsf{V}(\Gamma)}$.  Since their images are independent and 
$\Gamma_1 \cup \Gamma_2$ has exactly as many qubits as $\Gamma$, $\pi$ must be invertible. 

Next, we prove these choices respect the commutation relations as stated.
Consider $[X_{\tau(1)}]_1$: this error commutes with all the operators in Eq.~\eqref{eq:cedge}--\eqref{eq:xx-diag}
except $\pi(Z_1 Z_2) = [Z_{\tau(1)} ]_1$. There are $4\ell_f-3$ such hopping operators on 
$f$ with which $\pi^{-1}([X_{\tau(1)}]_1)$ must commute. As a consequence of the rank-nullity theorem there are $2^{{4\ell_f}-(4\ell_f-3)}$  such operators.
It can be verified that  $\langle X_1, B_f^X, B_f^Z\rangle $ account for these operators. But 
$\pi^{-1}([X_{\tau(1)}]_1)$ must also anti-commute with $Z_1Z_2$. This gives the choices in 
Eq.~\eqref{eq:z-split} since operators in $\langle  B_f^X, B_f^Z\rangle$ commute with $Z_1Z_2$.
Now let us determine $\pi^{-1}([X_{\tau(2m)}]_2)$. Once again with reference to 
Eq.~\eqref{eq:cedge}--\eqref{eq:xx-diag} we see that  it must commute with 
$4\ell_f-3$ hopping operators on $f$. It also commutes with $\pi^{-1}([X_{\tau(1)}]_1)$ 
since $[X_{\tau(2m)}]_2$ commutes with $[X_{\tau(1)]_1}$. Again, due to a dimensionality argument there are $2^{4\ell_f-(4\ell_f-2)}$ choices for $\pi^{-1}([X_{\tau(2m)}]_2)$. Since 
$[X_{2m}]_2$  anti-commutes with $[Z_{\tau(2m)}]_2$
its preimage must anti-commute with $\pi^{-1}([Z_{\tau(2m)}]_2)=X_{2m-1}X_{2m}$ giving  two choices $Z_{2m}$ and $Z_{2m}B_f^X$. We can check that $Z_{2m}$ satisfies all the required commutation relations as does the choice $Z_{2m} B_f^X$. 
\end{proof}

\noindent
In Lemma~\ref{lm:split} we first assigned $\pi^{-1}([X_\tau(1)]_1)$ followed by 
$\pi^{-1}([X_{\tau(2m)}]_2) $. Changing the order restricts $g$ to $\{I,B_f^X\}$ in Eq.~\eqref{eq:zz-diag} while $g\in \{I,B_f^X, B_f^Y, B_f^Z \}$ in Eq.~\eqref{eq:xx-diag}.

\begin{lemma}[Preserving commutation relations]\label{lm:commutation}
The map $\pi$ preserves commutation relations of error operators in $\mc{P}_{\mathsf{V}(\Gamma)}$. 
\end{lemma}

\begin{lemma}[Preserving code capabilities]\label{lm:stabilizers}
Under $\pi$, stabilizers of the color code on $\Gamma$ are mapped to stabilizers on the surface codes
on $\Gamma_1$ and $ \Gamma_2$. 
\end{lemma}

\renewcommand{\algorithmicrequire}{\textbf{Input:}}
\renewcommand{\algorithmicensure}{\textbf{Output:}}

\begin{algorithm}
\caption{{\ensuremath{\mbox{ Mapping a 2D color code to surface codes}}}}\label{alg:tcc-projections}
\begin{algorithmic}[1]
\REQUIRE {A 2-colex $\Gamma$ without parallel edges; $\Gamma$  is assumed to have a  2-cell embedding.}
\ENSURE { 
$\pi: \mc{P}_{\mathsf{V}(\Gamma)}\rightarrow \mc{P}_{\mathsf{E}(\Gamma_1)}\otimes \mc{P}_{\mathsf{E}(\Gamma_2)}$, where 
$\Gamma_i=\tau_c(\Gamma)$.}

\STATE Pick a color $c\in \{r,g,b\}$ and contract all edges of $\Gamma$ that are colored $\{r,g,b\}\setminus c$ to obtain $\tau(\Gamma)$. Denote two instances of $\tau(\Gamma)$ as $\Gamma_1$ and $\Gamma_2$.
\STATE Choose charges $\epsilon_{c}$, $\mu_{c}$, $\epsilon_{c'}$ and $\mu_{c'}$ on $\Gamma$, where $c' \neq c$.
\STATE Set up correspondence between charges on $\Gamma$ and $\Gamma_{i}$  as follows:
$\epsilon_1 \equiv \epsilon_c $,  $\mu_1 \equiv \mu_{c'}$, $\epsilon_2 \equiv \mu_c$
and $\mu_2 \equiv \epsilon_{c'}$ .
\FOR { each $c''$-colored face $f$ in $\mathsf{F}(\Gamma)$}
\STATE Let the boundary of $f$ be $v_1$, \ldots $v_{2{\ell_f}}$.
\STATE Choose a pair of $c'$-colored edges in  $\partial(f)$, say $(v_{2{\ell_f}},v_1)$ and $(v_{2m},v_{2m+1})$. Let $[ T  ]_i$  denote that $T$ acts on $\Gamma_i$. 
\begin{align}
\pi(Z_{v_1})  &= \left[ X_{\tau(v_1)}\right]_2 \prod_{i=1}^{m} \left[  Z_{\tau(v_{2i})} \right]_1  
\end{align}
\STATE For $1 \leq j\leq {\ell_f} $ compute the mapping (recursively) as 
\begin{align}
\pi(Z_{v_{2j}})  &=  \pi (Z_{v_{2j-1}}) \left[  Z_{\tau(v_{2j})}\right]_1\\
\pi(Z_{v_{2j-1}})  &= \pi (Z_{v_{2j-2}}) \left[X_{\tau(v_{2j-2})} X_{\tau(v_{2j-1})}\right]_2
\end{align}

\STATE For $1 \leq j\leq m$ compute the mapping as 
\begin{align}
\pi(X_{v_1})  &=  \left[ X_{\tau(v_1)}\right]_1\\
\pi(X_{v_{2j}})  &=  \pi (X_{v_{2j-1}}) \left[  Z_{\tau(v_{2j})}\right]_2\\
\pi(X_{v_{2j-1}})  &= \pi (X_{v_{2j-2}}) \left[ X_{\tau(v_{2j-2})}  X_{\tau(v_{2j-1})}\right]_1
\end{align}

\STATE For $m+1 \leq j\leq {\ell_f}$ compute the mapping as 
\begin{align}
\pi(X_{v_{2{\ell_f}}})  &=   \left[ X_{\tau(v_{2{\ell_f}})}\right]_1\\
\pi(X_{v_{2j-1}})  &= \pi (X_{v_{2j}}) \left[  Z_{\tau(v_{2j})}\right]_2\\
\pi(X_{v_{2j}})  &=  \pi (X_{v_{2j+1}}) \left[ X_{\tau(v_{2j})}  X_{\tau(v_{2j+1})}\right]_1
\end{align}
\ENDFOR
\end{algorithmic}
\end{algorithm}
\noindent

\begin{theorem}\label{th:tcc-map}
Any  2D color code (on a 2-colex $\Gamma$ without parallel edges) is equivalent to  a pair of surface codes $\tau(\Gamma)$ under the map $\pi$ defined as in  Algorithm~\ref{alg:tcc-projections}. 
\end{theorem}
{
\begin{proof}[Proof Sketch]
By charge conservation we require two copies of $\tau(\Gamma)$ to represent the color code using surface codes. 
Lines 2--3 follow from Lemma~\ref{lm:charge-pairing}. Since $c''$-colored faces in $\mathsf{F}_{c''}(\Gamma)$ cover all the qubits of the color code, we account for all the single qubit operators on the color code by the {$\mathsf {for}$}-loop in lines 4--10. 
The closed form expressions for single qubit errors in lines 6--9 are a direct consequence of Lemmas~\ref{lm:hopping}, \ref{lm:split} and the choices given in Eqs.~\eqref{eq:cedge}--\eqref{eq:xx-diag} and Eqs.~\eqref{eq:z-split}--\eqref{eq:x-split}.
 By considering the images of the stabilizers of the color code, we can show that they are mapped to the stabilizers of the surface codes on $\Gamma_i$ (see Lemma~\ref{lm:stabilizers}). From Lemma~\ref{lm:commutation},
  the commutation relations among the hopping operators on the color code in Eq.~\eqref{eq:cedge}--\eqref{eq:xx-diag} and the single qubit operators in Eq.~\eqref{eq:z-split}--\eqref{eq:x-split} are preserved. 
 Hence, the errors corrected by the color code are the same as those corrected by the surface codes on  $\Gamma_i$.
Thus the color code is equivalent to two copies of $\tau(\Gamma)$.
\end{proof}
}

\noindent{\em Acknowledgment.} This research was supported by the Centre for Industrial Consultancy \& Sponsored Research. 
We thank the referees for helpful comments and references.

{
\def\cprime{$'$}

}

\section*{appendix}
In this section we provide the proofs of Lemma~\ref{lm:commutation}~and~\ref{lm:stabilizers}.

\begin{proof}[Proof of Lemma~\ref{lm:commutation}]
We only sketch the proof. It suffices to show that the commutation relations hold for a basis of $\mc{P}_{\mathsf{V}(\Gamma)}$. We consider the basis consisting of the hopping operators along 
$c$ and $c'$ edges in Eq.~\eqref{eq:cedge}--\eqref{eq:xx-diag} and the single qubit operators given in Lemma~\ref{lm:split}. 
The proof of Lemma~\ref{lm:split} shows that the commutation relations are satisfied for the single qubit operators. Consider a hopping operator along  $c'$-colored edge. This anti-commutes with exactly two hopping operators along 
$c$-colored edges on $\Gamma$.  For instance consider $Z_{2i}Z_{2i+1}$. From Eq.~\eqref{eq:cedge}
--\eqref{eq:xx-diag}
this anti-commutes with $X_{2i-1}X_{2i}$ and $X_{2i+1}X_{2i+2}$. Their images under $\pi$ are 
$[X_{\tau(2i)}X_{\tau(2i+1)}]_2$, $[Z_{\tau(2i)}]_2$ and $[Z_{\tau(2i+1)}]_2$ for which it is clear that the commutation relations are satisfied. 

The operators along the $c$-colored edges are given in  Eq.~\eqref{eq:cedge}. Suppose we consider $\pi(Z_{2i-1}Z_{2i})$;
then it anti-commutes with $X_{2i-2}X_{2i-1}$ and $X_{2i}X_{2i+1}$. We only need to verify for those operators which are independent. Assume that they are both independent, then their images are $X_{\tau(2i-2)}X_{\tau(2i-1)}$ and $X_{\tau(2i)}X_{\tau(2i+1)}$ respectively. They anti-commute with $\pi(Z_{2i-1}Z_{2i}) = [Z_{\tau(2i-1)}]_1  = [Z_{\tau(2i)}]_1$.
If only one of the operators is independent, then we need only verify for that operator. The preceding argument already establishes this result. 
We can argue in a similar fashion to show that commutation relations are preserved for the operators of the type $X_{2i-1}X_{2i}$ and $X_{2i}X_{2i+1}$. 
\end{proof}

\begin{proof} [Proof of Lemma~\ref{lm:stabilizers}]
To prove this, it suffices to show that the stabilizers associated with plaquettes of all three colors are mapped to stabilizers on the surface codes. If $\Gamma_i=\tau_c(\Gamma)$, then we show that the stabilizers associated with  $f\in \mathsf{F}_{c'}(\Gamma)\cup \mathsf{F}_{c''}(\Gamma)$ are mapped to the plaquette stabilizers on $\Gamma_i$.
If  $f\in  \mathsf{F}_{c'}(\Gamma)$, then $B_f^Z = \prod_{i=1}^{\ell_f} H_{2i-1,2i}^{\epsilon_{c}}$. By Lemma~\ref{lm:hopping} this is mapped to $\prod_{i}^{\ell_f} [Z_{\tau(2i)}]_1  = \prod_{e\in \partial(\tau(f))}[Z_e]_1$. Using a similar argument we can show that $B_f^Z \in  \mathsf{F}_{c''}(\Gamma)$ is also a plaquette stabilizer on $\Gamma_1$. Since faces in $\mathsf{F}_{c'}(\Gamma)\cup \mathsf{F}_{c''}(\Gamma)$ are in one to one correspondence with the faces of $\tau(\Gamma)$, they account for all the face stabilizers on 
$\Gamma_1$. By considering $B_f^X$, we can similarly show that they map to the face stabilizers on $\Gamma_2$. 

Now consider a face $f\in \mathsf{F}_{c}(\Gamma)$. Consider $B_f^Z$, this can be decomposed into hopping operators $H_{u,v}^{\epsilon_{c'}}$ along $c'$-edges. By Lemma~\ref{lm:hopping}, such an operator maps to $[X_{\tau(u)}X_{\tau(v)}]_2$ and an additional stabilizer on one of the faces of $\Gamma_1$ if $H_{u,v}^{\epsilon_{c'}}$ is a dependent hopping operator. Thus  $B_f^Z$ maps to a vertex operator on $\tau(f)$ in $\Gamma_2$ and possibly a combination  of plaquette stabilizers. Since every vertex in $\Gamma_2$ is from a face in $\Gamma$, we can account for  all the vertex operators on $\Gamma_2$. Similarly, by considering the stabilizer $B_f^X$ we can account for all the vertex operators on $\Gamma_1$. 
\end{proof}

We illustrate our results with an example. We consider the color code on the hexagonal lattice, see Fig.~\ref{fig:hexmap}. The color code is shown on the left and the surface codes on the right. Each of the single qubit errors and their images are shown (in bold red). 
\begin{figure}
\centering
\includegraphics[scale=0.5]{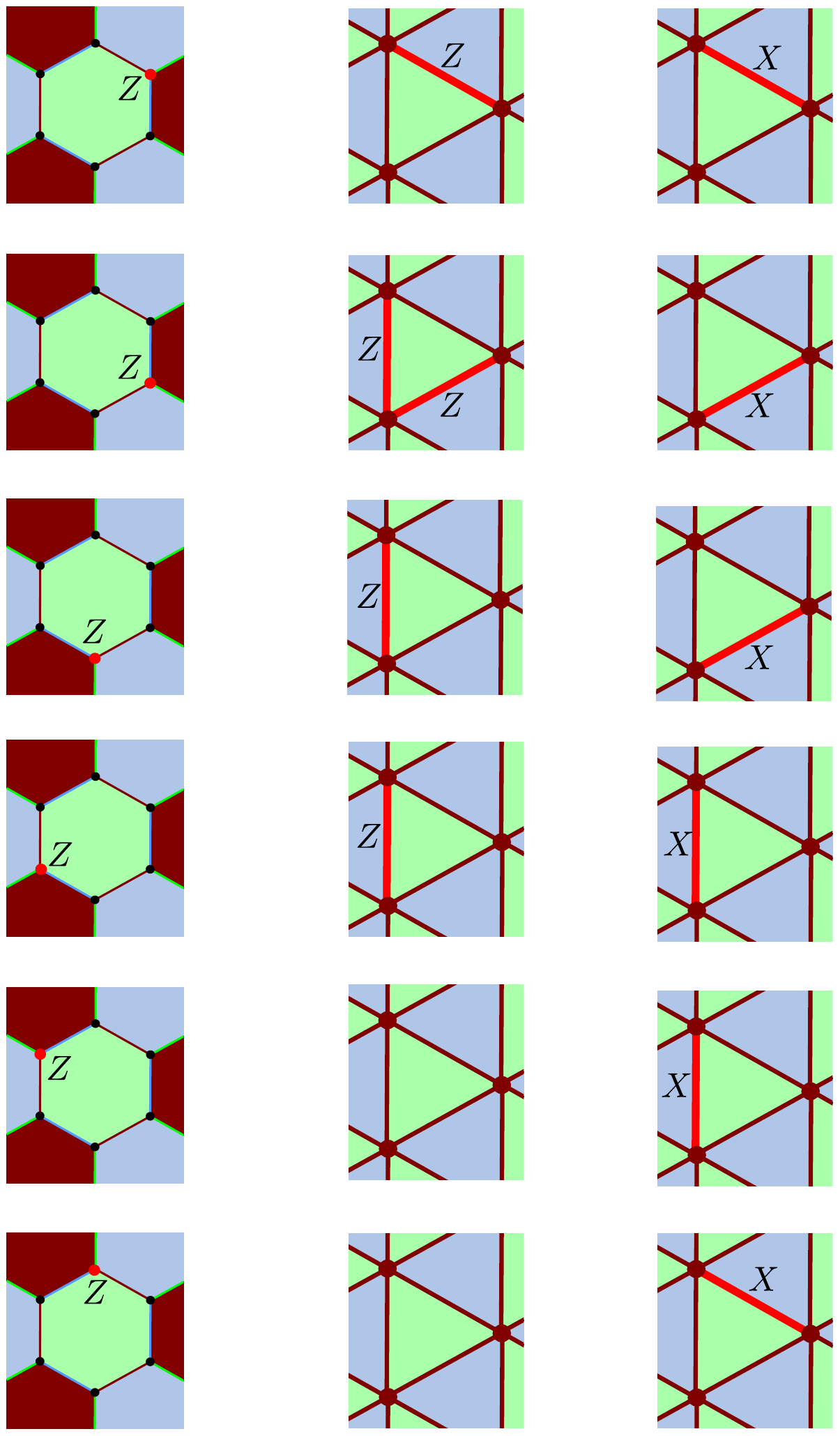}\\
\vspace{5mm}
\includegraphics[scale=0.5]{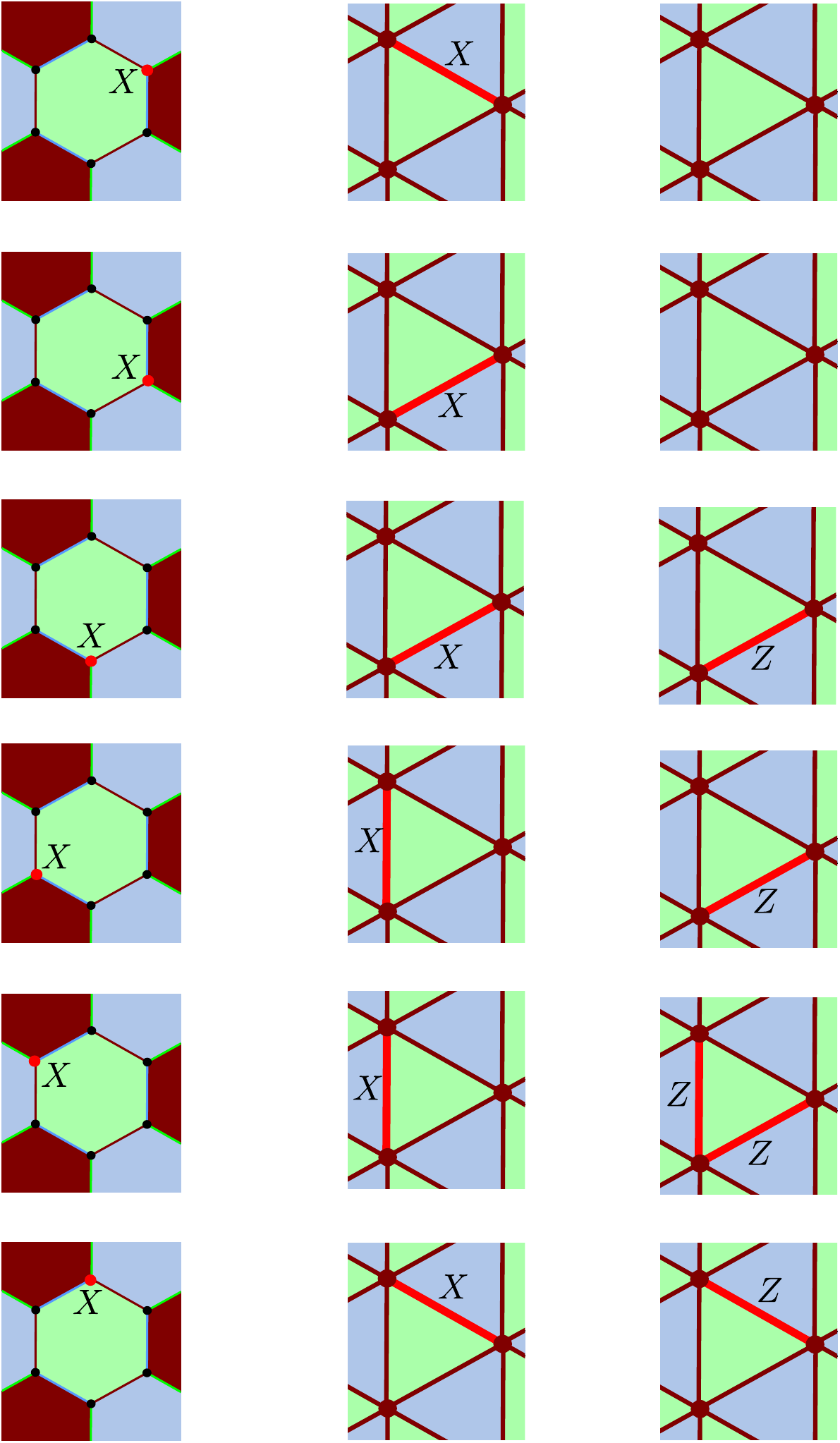}
\caption{Map for the hexagonal color code, showing the single qubit errors on color code and the associated images on the surface codes. }\label{fig:hexmap}
\end{figure}

\end{document}